%% file: ACC2016.tex
\documentclass[letterpaper, 10 pt, conference, english]{ieeeconf} % Comment this line out
\input{format}

 \title{Uniform-Price Mechanism Design for a Large Population of \\ Dynamic Agents }
  \author{Sen Li, Wei Zhang, Jianming Lian, and Karanjit Kalsi
 \thanks{S. Li, and W. Zhang are with the Department of Electrical and Computer Engineering, Ohio State University, Columbus, OH 43210. Email: \{li.2886, zhang.491\}@osu.edu}
  \thanks{J. Lian and K. Kalsi are with the Electricity Infrastructure Group, Pacific Northwest National Laboratory, Richland, WA 99354. Email: \{jianming.lian, karanjit.Kalsi\}@pnnl.gov}
 }

\begin{document}
\maketitle
\begin{abstract}
This paper focuses on the coordination of a large population of dynamic agents with private information over multiple periods. Each agent maximizes the individual utility, while the coordinator determines the market rule to achieve group objectives. The coordination problem is formulated as a dynamic mechanism design problem. A mechanism is proposed based on the competitive equilibrium of the large population game. We derive the conditions for the general nonlinear dynamic systems under which the proposed mechanism is incentive compatible and can implement the  social choice function in $\epsilon$-dominant strategy equilibrium. In addition, the proposed mechanism is applied to a linear quadratic problem with bounded parameters to show its efficacy.

\end{abstract}
\input{Introduction}
\input{ProblemStatement}

\input{Solution}

\input{Simulation}

\section{conclusion}
This paper presented a dynamic mechanism design approach for the coordinator of a large population agents with private information. A mechanism was proposed to motivate self-interested agents to achieve group objectives. We showed that the proposed mechanism is incentive compatible, and it implements the desired social choice function in $\epsilon$-dominant strategy equilibrium. We also applied the proposed mechanism to linear quadratic case to show its efficacy. Future work includes developing the mechanism design framework for more general dynamic models, and applying the results to practical applications such as thermostatically controlled loads, plug-in electric vehicles, dyers, washers, among others.

\bibliographystyle{unsrt}
\bibliography{Allerton2015}

\end{document}

%% file: format.tex
\IEEEoverridecommandlockouts
\usepackage{epsfig}
\usepackage{subfigure}
\usepackage{amssymb,amsmath,amsfonts,layout,graphicx}
\usepackage{makeidx}
\usepackage{babel}
\usepackage{sublabel}
\usepackage{cases}
\usepackage{algorithm}
\usepackage{algorithmic}

\usepackage
[
        letterpaper,% other options: a3paper, a5paper, etc
        left=1.91cm,
        right=1.91cm,
        top=1.91cm,
        bottom=2cm,
]
{geometry}

\newtheorem{proposition}{Proposition}

\newtheorem{definition}{Definition}
\newtheorem{corollary}{Corollary}
\newtheorem{remark}{Remark}

%% file: Introduction.tex
\section{Introduction}
Mechanism design deals with the coordination of a group of self-interested agents to achieve group objectives \cite{mas1995microeconomic}. 
In a mechanism design problem, the coordinator needs to provide proper incentives to align the individual preferences with the social choice, i.e., the strategic behaviors of the self-interested  agents result in an outcome that achieves the desired group objectives. Such problem has been extensively studied by researchers in both engineering and economics \cite{zhangtruthful}, \cite{samadi2012advanced},     \cite{mcsherry2007mechanism}, \cite{nissim2012approximately}, \cite{kearns2014mechanism}, \cite{azevedo2013strategy}. 
Some of these works consider the coordination problem in dynamic setting
\cite{bergemann2010dynamic}, \cite{balandat2013dynamic}, \cite{athey2013efficient}, \cite{pavan2014dynamic}, \cite{pavan2009dynamic}, \cite{mierendorff2011optimal}.
While these existing works guarantee incentive compatibility and realize efficient resource allocation, they typically lead to discriminatory pricing schemes: the unit price of the mechanism is different for different agents. Such pricing scheme is  not applicable in many applications, including the existing electricity market. 

On the other hand, uniform-price mechanisms draw extensive research attentions due to its wide applicability in engineering problems. Many works study the design of uniform-price incentive mechanisms for price-taking agents \cite{limarket}, \cite{bitar2014deadline},  \cite{papadaskalopoulos2013decentralized} who do not anticipate the effect of his bid on the market clearing price. While intuitively the effect of individual agent actions on market price may vanish if the population grows large, the price-taker assumption requires a more rigorous justification.  In the case where the agents are not treated as price-takers, the interactions between agent decisions make it challenging to even analyze the game solutions of a given mechanism. One branch of works studies the decentralized convergence to the solution of some given mechanisms using mean-field approximations  \cite{kohansal2014price}, \cite{huang2012social}, \cite{grammatico2015mean}, \cite{parise2015constrained}, \cite{grammatico2014decentralized}, assuming that the agents are indifferent to a change in his utility up to a constant. While such method proposes a way to efficiently compute the game solutions, it does not involve group objectives. In contrast, another branch of works designs mechanisms that take group objectives into account \cite{chen2010two}, \cite{xudemand}, \cite{wu2012vehicle}.  In these works, proper pricing schemes are proposed to motivate agents to achieve or approximately achieve socially optimal resource allocations. However, most of these works focus on static mechanism design, where system dynamics are not involved. 

In this paper we consider the coordination of a large population of dynamic agents to achieve group objectives over multiple periods. Each individual agent is self-interested and strategically chooses control actions to maximize his utility. In the meanwhile, the coordinator needs to provide proper incentives (typically via pricing) for individual agents to achieve certain group objectives, i.e., maximize the social welfare. This problem poses several challenges: first, we focus on uniform-price mechanisms, where the agents face a single-valued unit price. No existing tool in mechanism design is directly applicable in this scenario. Second, to identify the optimal mechanism that maximizes the social welfare, a subproblem is to determine the solution of the game induced by a given mechanism. This problem is challenging, especially when large-scale coupled dynamic systems are involved.  Existing works mostly focus on mean field control \cite{kohansal2014price}, \cite{huang2012social}, \cite{grammatico2015mean}, \cite{parise2015constrained}, \cite{grammatico2014decentralized}, where the coordinator's control is typically fixed as the average of individual states or controls. In this case, since the higher level control is fixed (not optimized), the resulting solution can not guarantee to achieve the group objectives. Therefore, a more general framework that addresses other control scenarios is needed. 

The key contribution of this work lies in identifying a class of problems where uniform price mechanism can be developed to implement he social choice function in $\epsilon$-dominant strategy equilibrium. This class of problems captures a wide range of resource allocation problems where the coordinator allocates resource to a large population of dynamic agents to maximize the social welfare subject to a peak resource constraint. Despite its wide applicability, no existing tools in mechanism design literature can be directly used to design a uniform-price mechanism for this particular problem. In this paper, we propose a uniform-price mechanism to achieve the group objectives. We show that when the impact of individual bid on the market clearing price is bounded, the proposed mechanism is incentive compatible for general nonlinear dynamic systems. We also derived conditions under which the proposed mechanism can implement the social choice function in $\varepsilon$-dominant strategy equilibrium. Furthermore, the proposed mechanism is applied to a linear quadratic cases. We show that as long as the system parameters are properly bounded, the conditions for incentive compatibility can be verified, and the proposed mechanism maximizes the social welfare in $\epsilon$-dominant strategy equilibrium. 

The rest of the paper proceeds as follows. The coordination problem is formulated as a dynamic mechanism design problem in Section II. A uniform-price mechanism is proposed in Section III. In Section IV a special case is studied to show the efficacy of the proposed approach.

%% file: ProblemStatement.tex
\section{Problem Formulation}
In this section we formulate the coordination problem as a mechanism design problem. Each agent is modeled as an individual utility maximizer who chooses a bid to optimize individual utility over a planning horizon. In the meanwhile, the coordinator needs to design the message space and market clearing rules (mechanism) to motivate the self-interested agents to achieve group objectives.

\subsection{Individual Utility Maximization}
Consider a group of $N$ agents with a planning horizon of length $K$. Let $x_k^i$ and $a_k^i$ denote the system state and control action of the $i$th agent at $k$th period, respectively. We consider the following discrete time dynamic systems with constraint:
\begin{align}
\label{dynamicsystems}
    \begin{cases}
	   x^i_{k+1}=f_i(x_k^i,a_k^i; \theta_k^i)  \\
	   a_k^i\in \Omega_k^i \\
	   x_k^i\in X_k^i
	\end{cases}
\end{align}
where $\theta_k^i\in \Theta^i$ denotes the private information of agent $i$, $\Omega_k^i$ and $X_k^i$ are convex and compact set, and $f_i$ is assumed to be continuous with respect to both parameters. In the mechanism design context, the private information is typically referred to as the agent type. In the individual decision problem, each agent acts strategically to maximize the individual utility, which is the valuation minus the payment. The user's valuation can be captured by a stage-additive valuation function $V_k^i:X_k^i\times \Omega_k^i \rightarrow R$, which is assumed to be continuously differentiable and strictly convex with respect to $a_k^i\in \Omega_k^i$. It evaluates the level of comfort the agent experiences when receiving an allocation decision $a_k^i\in \Omega_k^i$ at state $x_k^i \in X_k^i$. Let $p_k$ denote the unit price of the allocation at time $k$, then the utility of the $i$th agent can be denoted as follows:
\begin{equation}
\label{invidualutility}
	\sum_{k=1}^K U_k^i(x_k^i,a_k^i,p_k;\theta_k^i)=\sum_{k=1}^K V_k^i(x_k^i,a_k^i;\theta_k^i)-p_ka_k^i
\end{equation}
Since the agents are self-interested, each agent would choose control actions to maximize his utility subject to the dynamic constraint (\ref{dynamicsystems}). 
\begin{remark}
In the above formulation, the unit price $p_k$ is the same to different agents. This is called the non-discriminatory pricing. The non-discriminatory pricing scheme is compatible with many existing markets, and is easy to implement. On the other hand, many existing works are based on Vickrey-Clarke-Groves (VCG) mechanisms  \cite{mas1995microeconomic}, \cite{samadi2012advanced}, \cite{furuhata2015online}, which typically leads to discriminatory pricing. The uniform-price mechanism design problem (especially the dynamic case) is not well studied yet. 
\end{remark}

\subsection{The Mechanism Design Problem} 
Since the agents have private information, the coordinator asks each agent to submit a bid denoted as $r^i=(r_1^i,\ldots,r_L^i)$. This bid must be in the message space, $M^i$, specified by the coordinator. The coordinator then collects these bids and determines the market outcome accordingly, which consists of the allocation decision $a=(a_1,\ldots,a_K)$ and the unit price $p=(p_1,\ldots, p_K)$, where $a_k=(a_k^1,\ldots,a_k^N)$. This market clearing process maps the agent bids to the market outcome, which can be modeled as an outcome function $g:M^1\times \cdots \times M^N \rightarrow \Omega \times P$, where $\Omega$ denotes all feasible resource allocations and $P$ denotes all feasible unit prices.  A \emph{mechanism} consists of the message space and the outcome function, whose formal definition is given below:
\begin{definition}[\text{Mechanism} \cite{mas1995microeconomic}]
\label{mechanism}
A mechanism $\Gamma=(M,g(\cdot))$ consists of a collection of message spaces $M$ and an outcome function $g(\cdot)$,  where $M=M^1\times \cdots \times M^N$.
\end{definition}

Based on the definition, the outcome function can be written as $g(r)=(a^1,\ldots,a^N,p)$. For notation convenience, we denote $g_{k,a}^i(r)=a_k^i$ and $g_{k,p}(r)=p_k$, and let $g_a(r)=(a^1,\ldots,a^N)$. 
In our problem, when a mechanism is given, the utility function of each agent depends on the bid, and can be written as the following:
\begin{equation}
\label{inducedgame}
\sum_{k=1}^K U_k^i(x_k^i,g_{k,a}^i(r),g_{k,p}(r);\theta_k^i).
\end{equation}
As each agent is self-interested, he would choose his bid to maximizes the individual utility (\ref{inducedgame}). In the utility function (\ref{inducedgame}), since $r$ is the vector consisting of all agent bids, each agent's utility depends on the actions of other agents. Therefore, this utility maximization problem is a game induced by the mechanism $\Gamma$, which can be denoted as $G(\Gamma)$.  
For a game problem, there are several solution concepts, including Nash equilibrium, Bayesian Nash equilibrium, dominant strategy equilibrium, etc \cite{mas1995microeconomic}. In this paper we choose the $\epsilon$-dominant strategy equilibrium solution concept, which assumes that each agent is indifferent to a change in his utility up to a constant $\epsilon$. In reality, a utility of $0.99\$$ may have no difference from a utility of $1\$$ to an agent. Therefore, this assumption is reasonable in many practical applications. 
If we let $\Lambda_k^i(r)=(g_{k,a}^i(r),g_{k,p}(r))$, then the stage utility of the $i$th agent at time $k$ can be written as a function of $\Lambda_k^i(r)$ only (neglecting the dependence on initial state). Using these notations, the formal definition of the $\epsilon$-dominant strategy equilibrium is given as below:
\begin{definition}
Given a game $G(\Gamma)$, a bidding collection $r^*=(r^{1*},\ldots,r^{N*})$ is an $\epsilon$-dominant strategy equilibrium of the game if for any $i=1,\ldots,N$, we have:
\begin{equation}
\sum_{k=1}^K U_k^i(\Lambda_k^i(r^{i*},r^{-i}); \theta_k^i)\geq  \sum_{k=1}^K U_k^i(\Lambda_k^i(r^i,r^{-i}); \theta_k^i)-\epsilon
\end{equation}
for all $r^i\in M^i$ and all $r^{-i}\in M^{-i}$, where $M^{-i}=M^1\times\cdots\times M^{i-1}\times M^{i+1} \times \cdots \times M^N$.
\end{definition}
In the equilibrium solution, it is to the agent's best benefit to follow the equilibrium strategy regardless of other agent's actions. Note that the equilibrium strategy depends on user types, the individual decision making process can be captured by a bidding strategy $m_{\Gamma}^i:\Theta^i\rightarrow M^i$. In this notaion, $m_{\Gamma}(\theta)=(m_{\Gamma}^1,(\theta^1)\ldots,m_{\Gamma}^N(\theta^N))$ is the game solution to $G(\Gamma)$.
\vspace{0.1cm}
\begin{remark}
The dominant strategy equilibrium solution is a rather robust solution concept. In dominant strategy equilibrium, each agent's best decision is to follow the equilibrium strategy regardless of the actions of other agents. Therefore, when an individual agent makes a bidding decision, he does not need to anticipate the control actions of other agents. However, it is well-known that in general, VCG mechanism is the unique mechanism that induces dominant strategy equilibrium \cite{mas1995microeconomic}, which does not guarantee uniform unit price. Therefore, different from many existing works, we adopt the $\epsilon$-dominant strategy equilibrium. In this solution concept, a uniform-price mechanism (different from VCG) is proposed to implement the desired social choice function for large-population agents. 
\end{remark}

The goal of this paper is to find the mechanism that achieves the group objective. The group objective can be encoded in a social choice function $\phi:\Theta^1\times \cdots \times \Theta^N \rightarrow \Omega$ that maps the agent type to a feasible resource allocation. Here we consider the social choice function that maximizes the social welfare subject to a total resource constraint:
\begin{align}
\label{socialchoice}
\phi(\theta)=\arg &\max_{a} \sum_{i=1}^N \sum_{k=1}^K V_k^i(x_k^i,a_k^i;\theta_k^i)-\sum_{k=1}^K \sigma_k\bigg(\sum_{i=1}^N a_k^i\bigg) \nonumber \\
&\text{s.t. }\begin{cases}
x^i_{k+1}=f_i(x_k^i,a_k^i; \theta_k^i)  \\
a_k^i\in \Omega_k^i  \\
\sum_{i=1}^N a_k^i\leq D_k \quad \forall k \nonumber
\end{cases} 
\end{align}
In the above definition, $D_k$ is the total resource constraint, and $\sigma_k(\cdot)$ is a convex and continuously differentiable function denoting the cost for the coordinator to procure certain amount of resources at time $k$. 
The value of the social choice function indicates the socially desired resource allocation. On the other hand, the resource allocation is a function of the user bids, which depends on the mechanism. Therefore, to achieve the group objective, the coordinator needs to find a mechanism whose resulting allocation plan coincides with the value of the social choice function.
\begin{definition}[\text{Implementation} \cite{mas1995microeconomic}]
\label{implementation}
The mechanism $\Gamma$ implements the social choice function $\phi(\cdot)$ if there exists an equilibrium strategy profile $(m^{1}_{\Gamma}(\cdot), \ldots, m^{N}_{\Gamma}(\cdot))$ of the game $G(\Gamma)$ such that 
\begin{equation}
g_{a}(m^1_{\Gamma}(\theta_1), \ldots, m^{N}_{\Gamma}(\theta_N))=\phi(\theta), \quad \forall \theta\in \Theta.
\label{impl}
\end{equation}
\end{definition}
\vspace{0.2cm}
where $\Theta=\Theta^1\times \cdots \times \Theta^N$.

To this point, the problem of this paper can be stated as follows:
design a uniform-price mechanism to implement the social choice function $\phi(\cdot)$ in $\epsilon$-dominant strategy equilibrium..

%% file: Solution.tex
\section{The Uniform-price Mechanism}
In this section, we first propose a mechanism based on the competitive equilibrium of the market. Then, we show that the proposed mechanism can implement the social choice function when each agent strategically chooses bid to maximize individual utilities. Due to revelation principle, we focus on the direct mechanism, where the message space is the space of agent types, i.e., $M^i=\Theta^i$ for $\forall i$.

\subsection{The Proposed Mechanism}
To introduce our proposed mechanism, let us first consider a price-response problem, where each agent  takes price as given and chooses optimal control $a^i$ to maximize the individual utility. In the price-response problem, the control strategy of each agent depends on the price. Therefore, we have the following:
\begin{align}
\label{priceresponse}
\mu^i(p,r^i)=\arg &\max_{a^i} \sum_{k=1}^K U_k^i(x_k^i,a_k^i,p_k;r_k^i)\\
&\text{s.t. }\begin{cases}
x^i_{k+1}=f_i(x_k^i,a_k^i; r_k^i)  \\
a_k^i\in \Omega_k^i. \nonumber
\end{cases} 
\end{align}
where $u^i$ is assumed to be continuous with respect to $p$. In the context of economics, the price-allocation pair $(p,\mu^1(p,r^i),\ldots,\mu^N(p,r^i))$ is defined as the competitive equilibrium \cite{mas1995microeconomic}. Here we propose a mechanism $\Gamma_c$ based on the competitive equilibrium concept to solve the mechanism design problem of this paper.
\begin{definition}
\label{mechanism}
The proposed mechanism $\Gamma_c$ is a direct mechanism with the following outcome function:
\begin{align}
\begin{cases}
g_p(r)=\arg \max_{p} \sum_{i=1}^N \sum_{k=1}^K V_k^i(x_k^{i*},a_k^{i*};r_k^i)- \\
\quad \quad \quad \quad \quad \quad \quad \sum_{k=1}^K \sigma_k\big(\sum_{i=1}^N a_k^{i*}\big)  \nonumber \\
\quad \quad \quad \quad \quad  \text{s.t. }\begin{cases}
x^{i*}_{k+1}=f_i(x_k^{i*},a_k^{i*}; r_k^i)  \\
a^{i*}=\mu^i(p,r^i) \\
\sum_{i=1}^N a_k^{i*}\leq D_k, \quad \forall k \\
a_k^{i*}\in \Omega_k^i   \\
\end{cases} \\
g_a^i(r)=\mu^i(g_p(r),r^i).
\end{cases}
\end{align}
\end{definition}
\vspace{0.1cm}
where $g_a^i(r)=(g_{1,a}^i(r),\ldots,g_{K,a}^i(r))$ and $g_p(r)=(g_{1,p}(r),\ldots,g_{K,p}(r))$. In the above definition, the utility function is a continuous function with respect to $p$ over a compact set. Therefore, $g_p(r)$ exists, and the proposed mechanism is well-defined. 
\vspace{0.2cm}

\begin{remark}
In this paper, we assume that $\mu^i$ is continuous with respect to $p$, which establishes the existence of $g_p(r)$. It is well known that the solution to a strictly convex parametric quadratic programming problem is continuous with respect to its parameters \cite{hempel2015inverse}. Therefore, this continuity assumption holds when $U_k^i$ is quadratic and $f_i$ is linear. We believe this assumption also holds for more general systems, but the derivation of the detailed conditions that guarantee the continuity of $\mu^i$ is out of the scope of this paper. 
\end{remark}

\begin{remark}
Note that although the mechanism is defined based on the price-taker game (\ref{priceresponse}), we do not assume that agents are price-taker. In our problem, each agent will anticipate the impact of his bid on the market clearing price, and makes bidding decisions accordingly. This is significantly different from the works based on competitive equilibrium \cite{limarket}, \cite{chen2010two}, \cite{xudemand}. On the other hand, we conjecture that the impact of the bids of an individual agent on the market price will diminish as the population size grows. In this case the agent behavior should be very similar to price-taker agents when the game has a large population. The rest of this section provides rigorous justification for this conjecture. 
\end{remark}

\subsection{Properties of the Proposed Mechanism}
This subsection discusses some properties of the proposed mechanism. 

In mechanism design problems, the coordinator typically wants to design a mechanism to induce the truth-telling solution ($r^i=\theta^i$). If truth-telling is an equilibrium strategy of the game induced by a direct mechanism, we say this mechanism is incentive compatible. The formal definition of incentive compatibility is as follows:
\begin{definition}[\text{Incentive Compatibility} \cite{mas1995microeconomic}]
A direct mechanism $\Gamma$ is incentive compatible in $\epsilon$-dominant strategy equilibrium if for any $i=1,\ldots,N$, we have:
\begin{equation}
\sum_{k=1}^K U_k^i(\Lambda_k^i(\theta^{i},r^{-i}); \theta_k^i)\geq  \sum_{k=1}^K U_k^i(\Lambda_k^i(r^i,r^{-i}); \theta_k^i)-\epsilon
\end{equation}
for all $r^i\in M^i$ and all $r^{-i}\in M^{-i}$.
\end{definition}
\vspace{0.1cm}

Under certain conditions, we can show that the proposed mechanism is incentive compatible in $\epsilon$-dominant strategy equilibrium. To present our result, we first introduce some regularity conditions.
\begin{definition}
We say a function $\varphi:R^k\rightarrow R$ is Lipschitz continuous with constant $C$ in $L_{\infty}$ norm for $x\in X$, if there exists $C>0$ such that:
\begin{equation}
\label{Lipschitz}
|\varphi(x_1)-\varphi(x_2)|\leq C||x_1-x_2||_{\infty}, \forall x_1\in X, x_2\in X
\end{equation}
\end{definition}
\vspace{0.2cm}

Throughout the rest of the paper, whenever we say a function is Lipschitz continuous with constant $C$, we mean the function is Lipschitz continuous with constant $C$ in $L_{\infty}$ norm for its parameters in the corresponding space. 

The incentive compatibility result of the proposed mechanism can be summarized as the following proposition:
\begin{proposition}
\label{incentivecomp}
Assume that $U_k^i$ is Lipschitz continuous with respect to both $a_k^i$ and $p_k^i$ with constant $C_1$ and $C_2$, respectively. Assume $\mu^i$ is Lipchitz continuous with respect to $p$ with constant $C_3$. Furthermore, assume that there exists $\epsilon_1>0$ for the proposed mechanism $\Gamma_c$ such that:
\begin{equation}
\label{condition1}
|g_{k,p}(r^i,r^{-i})-g_{k,p}(\theta^i,r^{-i})|\leq \epsilon_1,\quad \forall r^i\in M^i, r^{-i}\in M^{-i},
\end{equation}
then $\Gamma_c$ is incentive compatible in $\epsilon$-dominant strategy equilibrium where $\epsilon=K(C_1C_3+C_2)\epsilon_1$.
\end{proposition}
\begin{proof}
Let $r^i$ denote the bids of the $i$th agent. Let $p_k^*=g_{k,p}(r)$. According to the definition of the proposed mechanism $\Gamma_c$, for each agent $i$, we have the following:
\begin{equation}
\label{proof1eq}
\sum_{k=1}^K U_k^i\left(\mu_k^i(p_k^*,r^i),p_k^*;r^i\right)\geq \sum_{k=1}^K U_k^i\left(\mu_k^i(p_k^*,\tilde{r}^i),p_k^*;r^i\right) 
\end{equation}
for all $\tilde{r}^i\in M^i$. Due to the Lipschitz continuity and the assumption (\ref{condition1}) of the proposition, it can be verified that the following inequality holds:
\begin{align}
\label{proof2eq}
&\sum_{k=1}^K U_k^i\left(\mu_k^i(g_p(r^i,r^{-i}),\tilde{r}^i),g_{k,p}(r^i,r^{-i});r^i\right)+\epsilon\geq \nonumber \\
& \sum_{k=1}^K U_k^i\left(\mu_k^i(g_p(\tilde{r}^i,r^{-i}),\tilde{r}^i),g_{k,p}(\tilde{r}^i,r^{-i});r^i\right) 
\end{align}
for all $\tilde{r}^i\in M^i$ and $r^{-i}\in M^{-i}$, where $\epsilon=K(C_1C_3+C_2)\epsilon_1$. Equation (\ref{proof1eq}) and (\ref{proof2eq}) together indicate the following:
\begin{align}
\label{proof3eq}
&\sum_{k=1}^K U_k^i\left(\mu_k^i(g_p(r^i,r^{-i}),r^i),g_{k,p}(r^i,r^{-i});r^i\right)\geq \nonumber \\
& \sum_{k=1}^K U_k^i\left(\mu_k^i(g_p(\tilde{r}^i,r^{-i}),\tilde{r}^i),g_{k,p}(\tilde{r}^i,r^{-i});r^i\right)-\epsilon 
\end{align}
for all $\tilde{r}^i\in M^i$ and $r^{-i}\in M^{-i}$. Therefore, the proposed mechanism guarantees that the agent bid $r$ satisfies the above condition (\ref{proof3eq}). 

On the other hand, the $i$th agent seeks to find a bid $r_i$ such that:
\begin{align}
\label{proof4eq}
&\sum_{k=1}^K U_k^i\left(\mu_k^i(g_p(r^i,r^{-i}),\theta^i),g_{k,p}(r^i,r^{-i});\theta^i\right)\geq \nonumber \\
& \sum_{k=1}^K U_k^i\left(\mu_k^i(g_p(\tilde{r}^i,r^{-i}),\tilde{r}^i),g_{k,p}(\tilde{r}^i,r^{-i});\theta^i\right)-\epsilon 
\end{align}
for all $\tilde{r}^i\in M^i$ and $r^{-i}\in M^{-i}$. Letting $r_i=\theta_i$ exactly satisfies this objective. This completes the proof. 
\end{proof}
\vspace{0.2cm}

Proposition \ref{incentivecomp} provides conditions under which the proposed mechanism is incentive compatible. Aside from Lipschitz continuity, the key condition is (\ref{condition1}): the bid of a particular agent will affect the market clearing price by at most $\epsilon_1$. While the result of Proposition \ref{incentivecomp} works for general nonlinear dynamic systems, how to verify condition (\ref{condition1}) depends on the particular application. In the next section we will show how this condition can be verified for linear quadratic problems. The verification of this condition for more general systems is left for future work.

%\begin{remark}
%We comment that in general it is rather challenging to characterize the solution to a large-scale dynamic game. Our result borrows some principle ideas from mean-field approximations  \cite{kohansal2014price}, \cite{huang2012social}, \cite{grammatico2015mean}, \cite{parise2015constrained}, \cite{grammatico2014decentralized}. However, these works are mostly based on mean-field control, which is not optimal for the desired social choice. In our problem, the optimal mechanism is much more complicated than mean-field control, and therefore more difficult to analyze. 
%\end{remark}

In our paper, since the mechanism is designed based on the competitive equilibrium, which deviates from the true agent behavior, incentive compatibility does not necessarily guarantee the group objective can be achieved. To account for this, we first introduce the following definition before proceeding to the key result of this paper. 
\begin{definition}
Given a resource allocation $a$, we say $a$ is uniform-price implementable if there exists a price $p\in P$ such that $a^i=\mu^i(p,\theta^i)$ for all $i=1,\ldots,N$.
\end{definition}
Based on the definition, if a resource allocation $a$ is uniform-price implementable by $p$, then $(a,p)$ is a competitive equilibrium of the game. 
Now we can present the key result of this paper:
\begin{proposition}
\label{keyresult}
Assume that the social choice function $\phi(\theta)$ is uniform-price implementable, and the proposed mechanism $\Gamma_c$ is incentive compatible in $\epsilon$-dominant strategy equilibrium, then $\Gamma_c$ can implement the social choice function $\phi(\theta)$ in $\epsilon$-dominant strategy equilibrium.
\end{proposition}
\begin{proof}
On one hand, based on Definition \ref{mechanism}, the proposed mechanism $\Gamma_c$ renders an optimal resource allocation among all uniform-price implementable allocations. Therefore, the resulting social welfare is no less than that of the desired social choice function $\phi(\theta)$.

On the other hand, the social choice function is the team problem for the proposed mechanism $\Gamma$ \cite{basar1995dynamic}, \cite{bacsar1980team}. Therefore, it provides an upper bound for attainable social welfare for $\Gamma$. This completes the proof. 
\end{proof}

%% file: Simulation.tex
\section{The Linear Quadratic Case}
In this section, we derive the conditions under which the proposed mechanism can achieve the group objective in linear quadratic problems.

Based on Proposition \ref{keyresult}, to achieve the group objective, the following conditions need to be verified: first, the social choice function should be uniform-price implementable; second, the propose mechanism $\Gamma_c$ needs to be incentive compatible. In the rest of this section, we will introduce the problem setup and check these two conditions in this setup. 

\subsection{Problem Setup}
Let us consider the following linear scalar system for the $i$th agent:
\begin{equation}
\label{linearmodel}
x_{k+1}^i=A^ix_k^i+B_k^ia_k^i
\end{equation}
where $x_k^i\in R$ denotes the state and $a_k^i\in R$ is the input. We assume that $0<\underline{A}\leq A_i \leq \bar{A}$ and $\underline{B}\leq B_i \leq \bar{B}<0$ for all $i=1,\ldots,N$. This linear model captures various dynamic assets, including batteries, PEVs, thermostatically controlled loads (with proper control strategy), among others. 

Let $(d_1^i,\ldots,d_K^i)$ be the desired state trajectory for agent $i$. Consider a valuation function of quadratic form that penalizes the deviation from this desired sequence. The individual utility can be written as follows:
\begin{equation}
\label{lqutility}
\sum_{k=1}^K U_k^i(x_k^i,a_k^i,p_k;\theta_k^i)=\sum_{k=1}^K \beta_k^i(x_k^i-d_k^i)^2-p_ka_k^i
\end{equation}
where $\theta_k^i=(A^i,B^i,\beta^i,d^i)$. Note that $B^i$ denotes a vector, and the same applies to other time-varying parameters.  

In this setup, the social choice function has the following form:
\begin{align}
\label{socialchoice1}
\phi_{LQ}(\theta)=\arg &\max_{a} \sum_{i=1}^N \sum_{k=1}^K \beta_k^i(x_k^i-d_k^i)^2-\sum_{k=1}^K \sigma_k\big(\sum_{i=1}^N a_k^i\big) \\
&\text{s.t. }\begin{cases}
x_{k+1}^i=A^ix_k^i+B_k^ia_k^i  \\
\sum_{i=1}^N a_k^i\leq D_k \quad \forall k \nonumber
\end{cases} 
\end{align}
where $\beta_k^i<0$ for all $i$ and $k$. For sake of analysis, we assume that the function $\sigma_k$ has a linear structure, i.e., $\sigma_k(\sum_{i=1}^N a_k^i)=p_k^w(\sum_{i=1}^N a_k^i)$. In electricity market, $p_k^w$ is the wholesale energy price.

\subsection{Uniform-price Implementability}
In this subsection we show that the social choice function $\phi_{LQ}(\cdot)$ is uniform-price implementable. 
\begin{proposition}
\label{uniformim}
The social choice function $\phi_{LQ}(\cdot)$ is uniform-price implementable, i.e., for any report $r$, there exists a price $p^c$ such that $\phi_{LQ}^i(r)=\mu^i(p_c,r^i)$ for $\forall i=1,\ldots,N$, where $\phi_{LQ}^i(r)$ is  the allocation to the $i$th agent. 
\end{proposition}
\begin{proof}
First, under a given price $p^c$, the optimality conditions for the price-response problem (\ref{priceresponse}) are as follows:
\begin{equation}
\label{priceresponse1}
-\sum_{k=1}^K \frac{\partial V_k^i(x_k^i)}{\partial a_k^i}\Big|_{a_k^i=\mu_k^i(p^c,r^i)}+p_k^c=0,
\end{equation}
On the other hand, the social welfare maximization problem (\ref{socialchoice1}) is convex. The necessary and sufficient optimality condition is that there exists a unique multiplier $\xi^k$, such that:
\begin{equation}
\label{team}
-\sum_{k=1}^K \frac{\partial V_k^i(x_k^i)}{\partial a_k^i}\Big|_{a_k^i=\phi_{LQ}^{i,k}(r)}+p^w_k+\xi_k=0,
\end{equation}
where $\phi_{LQ}^{i,k}(r)$ is the allocation to the $i$th agent at time $k$, and $\xi_k$ is non-negative and satisfies 
$\xi_k(\sum_{i=1}^N \phi_{LQ}^{i,k}(r)-D)=0$ for all $k=1,\ldots,K$.
Let $p^c_k=p_k^w+\xi_k$. It can be verified that $\phi_{LQ}^i(r)=\mu^i(p^c,r^i)$ is a solution to (\ref{priceresponse1}). Since the price-response problem is convex, optimality condition (\ref{priceresponse1}) is both necessary and sufficient. This indicates that $p^c_k=p_k^w+\xi_k$ can implement the social choice function, which completes the proof. 
\end{proof}
The result of Proposition \ref{uniformim} shows that there exists a price to implement the social choice function $\phi_{LQ}(\cdot)$. In addition to this, we can also derive the conditions for this price, which can be summarized as the following corollary:
\begin{corollary}
The price $p^c$ satisfies $\phi_{LQ}^i(r)=\mu^i(p_c,r^i)$ for all $i=1,\ldots,N$ if and only if the following condition holds for all $k=1,\ldots,K$:
\begin{align}
\begin{cases}
\tilde{p}^c_k=p^w_k  & \text{ if } \sum_i^N \mu_k^i(p^c,r^i)< D_k\\
\tilde{p}^c_k\geq p^w_k  & \text{ if } \sum_i^N\mu_k^i(p^c,r^i)=D_k
\end{cases}
\label{eq:condition1}
\end{align}
\end{corollary}
This result can be directly derived from the proof of Proposition \ref{uniformim}. We comment that the result holds when $\sigma_k(\cdot)$ is nonlinear. In this case, $p^w_k$ is the marginal cost for obtaining certain amount of resources.

\subsection{Incentive Compatibility}
Now we verify the incentive compatibility of the proposed mechanism based on Proposition \ref{incentivecomp}. First, we observe that under the linear quadratic setup, utility function $U_k^i$ is Lipschitz continuous with respect to both $a_k^i$ and $p_k^i$ for a bounded domain. To verify the Lipschitz continuity of the function $\mu^i$, we observe that the analytic expression of $\mu^i$ can be derived for the linear quadratic case: 
\begin{align}
\begin{cases}
\mu_k^i(p^c,r^i)=&\hspace{-0.3cm} \theta_k^ip_k^c+\theta_{k,k-1}^ip_{k-1}^c+\theta^i_{k,k+1}p^c_{k+1}, \quad \forall k\neq 1 \nonumber \\
\mu^i_1(p^c,r^i)=&\hspace{-0.3cm} \theta^i_1p^c_1+\theta^i_2p^c_2 
\end{cases}
\end{align}
where $\theta^i_1=\frac{A^i}{2\beta^i_1B_1^iB_1^i}$, $\theta^i_2=-\frac{A^i}{2\beta^i_{1}B_1^iB_2^I}$, and the coefficients for $k\neq 1$ are defined as follows: 
\begin{align}
\begin{cases}
\theta^i_k=\dfrac{1}{2\beta^i_kB^i_kB^i_k}+\dfrac{A^iA^i}{2\beta^i_{k-1}B^i_kB^i_k} \\
\theta^i_{k,k-1}=-\dfrac{A^i}{2\beta^i_{k-1}B^i_{k-1}B_k^i}\\
\theta^i_{k,k+1}=-\dfrac{A^i}{2\beta^i_{k}B^i_{k}B^i_{k+1}},
\end{cases}
\label{eq:coefficients}
\end{align}
Since the function $\mu^i$ is linear with respect to $p^c$, it is Lipschitz continuous. 
In addition to Lipschitz continuity, we can also show that condition (\ref{condition1}) holds for the linear quadratic case. 
\begin{proposition}
In the linear quadratic setup, there exists $C_4>0$ (not depending on $N$) for the proposed mechanism $\Gamma_c$ such that:
\begin{equation}
\label{condition2}
|g_{k,p}(r^i,r^{-i})-g_{k,p}(\theta^i,r^{-i})|\leq C_4/N,\quad \forall r^i\in M^i, r^{-i}\in M^{-i},
\end{equation}
\end{proposition}
\begin{proof}
To prove this result, we need to study how agent bid affects the price. The price and the agent bids satisfy (\ref{eq:condition1}). This optimality condition can be transformed into the following equation $\nu(p,r^i)=0$, where $p$ is the root of the mapping $\nu: P\times M^i\rightarrow P$, and this mapping can be defined as follows: 
\begin{align}
\label{mapping}
\nu_k(p,r^i)=\begin{cases}
\alpha_k^1(\sum_{i=1}^N\mu_k^i(p,r^i)-D_k), \quad \text{if } p_k\geq \sigma'(D_k)  \\
\alpha_k^2(p_k-\sigma'(\sum_{i=1}^N\mu_k^i(p,r^i))), \quad\text{otherwise}
\end{cases}
\end{align} 
where $\nu_k$ is the $k$th element of vector $\nu$. 

Let $J\nu_p$ and $J\nu_r$ be the Jacobian matrix of $\nu$ with respect to $p$ and $r^i$, respective. If $J\nu_p$ is invertible,  then $p$ is a function of $r^i$ (if not, $r^i$ does not affect $p$, and the result of the proposition holds). Let $Jp_r$ be the Jacobian matrix of $p$ with respect to $r^i$ and we have $Jp_r=-(J\nu_p)^{-1}J\nu_r$, whenever $J\nu_p$ exists. Based on (\ref{eq:coefficients}), since the report $r^i $ are bounded, each entry of the matrix $J\nu_r$ is bounded, and there exists $C_5>0$ and $C_6>0$ such that $C_5N\leq \big|\dfrac{\partial \nu_k}{\partial p_{k'}}\big|\leq C_6N$ for all $k$ and $k'$ where $p_k\geq \sigma'(D)$. Due to the particular structure of $J\nu_p$, we can derive the analytic form for $(J\nu_p)^{-1}$ and verify that each entry of this matrix is also bounded. Therefore, there exists $C_7$ not depending on $N$ such that $\big|\dfrac{\partial p_k}{\partial r_l^i}\big|\leq C_7/N$. Since $r^i$  is bounded, this indicates that there exists $m>0$ not depending on $N$, such that:
\begin{equation}
|g_{k,p}(r^i,r^{-i})-g_{k,p}(\theta^i,r^{-i})| \leq \sum_{l=1}^L \big| \dfrac{\partial p_k}{\partial r_l^i} dr_l^i\big| \leq \dfrac{LmC_7}{N},
\end{equation} 
where $L$ is the dimension of the user bid $r^i$. In this case, $C_4=LmC_7$, which completes the proof. 
\end{proof}
\vspace{0.1cm}

To this point, all the conditions for Proposition \ref{incentivecomp} are verified. Therefore, the proposed mechanism is incentive compatible. Furthermore, since the social choice function $\phi_{LQ}$ is uniform-price implementable, we conclude that the mechanism $\Gamma_c$ implements the social choice function $\phi_{LQ}$ in $\epsilon$-dominant strategy equilibrium.